\documentclass[10pt,journal,compsoc]{IEEEtran}
\ifCLASSOPTIONcompsoc
  \usepackage[nocompress]{cite}
\else
  \usepackage{cite}
\fi
\ifCLASSINFOpdf
\else
\fi

\usepackage{tikz} 
\usetikzlibrary{positioning}
\usetikzlibrary{arrows,decorations.pathmorphing,backgrounds,fit,petri}
\usepackage[autostyle]{csquotes}
\usepackage{placeins}
\usepackage{graphicx}
\usepackage{subcaption}
\usepackage[justification=centering]{caption}
\usepackage{amsthm}
\newtheorem{theorem}{Theorem}
\newtheorem*{remark}{Remark}
\usepackage{amsmath}
\usepackage{enumitem}
\usepackage{multirow}
\usepackage{subcaption}
\usepackage{algorithm}% http://ctan.org/pkg/algorithm
\usepackage{algpseudocode}% http://ctan.org/pkg/algorithmicx
\usepackage{epstopdf}
\DeclareMathOperator*{\argmax}{argmax} % thin space, limits underneath in displays
\algrenewcommand\algorithmicrequire{\textbf{Input:}}
\algrenewcommand\algorithmicensure{\textbf{Output:}}

\SetLabelAlign{LeftAlignWithIndent}{\hspace*{1.0ex}\makebox[1.25em][l]{#1}}

\begin{document}

\title{Optimal Divisible Load Scheduling for Resource-Sharing Network}

\author{Fei~Wu,
        Yang~Cao,
        and~Thomas~Robertazzi,~\IEEEmembership{Fellow,~IEEE}% <-this % stops a space
\IEEEcompsocitemizethanks{\IEEEcompsocthanksitem Fei Wu, Yang Cao and T. Robertazzi are with the Department
of Electrical and Computer Engineering, Stony Brook Univesity, Stony Brook,
NY, 11794.\protect\\
E-mail: \{fei.wu,yang.cao,thomas.robertazzi\}@stonybrook.edu}}
\IEEEtitleabstractindextext{
\begin{abstract}
Scheduling is an important task allowing parallel systems to perform efficiently and reliably. For modern computation systems, divisible load is a special type of data which can be divided into arbitrary sizes and independently processed in parallel. Such loads are commonly encountered in applications which are processing a great amount of similar data units. For a multi-task processor, the processor's speed may be time-varying due to the arrival and departure of other background jobs. This paper studies an optimal divisible loads scheduling problem on a single level tree network, whose processing speeds and channel speeds are time-varying. Two recursive algorithms are provided to solve this problem when the arrival and departure times of the background jobs are known a priori and an iterative algorithm is provided to solve the case where such times are not known. Numerical tests and evaluations are performed for these three algorithms under different numbers of background jobs and processors.
\end{abstract}

}

\maketitle
\IEEEdisplaynontitleabstractindextext
\IEEEpeerreviewmaketitle

\IEEEraisesectionheading{\section{Introduction}\label{sec:introduction}}
\subsection{Background}
\IEEEPARstart{T}{he} design of efficient load scheduling algorithms has long been a pivotal concern in parallel processing applications. A parallel system refers to all classes of parallel computers from multicore CPUs to wide area computational grids comprising distributed and heterogeneous installations owned by mutually unrelated institutions [1]. A schedule is an assignment of tasks to processors in time. Parallel systems cannot be fully utilized if the applications are not properly scheduled. In modern networked systems, scheduling becomes more crucial due to the increasing prevalence of data-intensive computing. To deal with the large amount of data in modern computation system, divisible load theory (DLT) has emerged as a potential tool.\par
DLT assumes that computation and communication loads can be divided into parts of arbitrary sizes, which can be processed independently in parallel [2]. There are two assumptions for the loads in DLT. First is arbitrary divisibility and second is independence of execution. Loads may be divisible in fact or as an approximation. Such loads are commonly encountered in applications which are processing great amount of similar data units, such as image processing, signal processing, processing of massive experimental data, and so on [3]. In classic DLT models, there is usually a control processor holding all the data originally and then one can distribute such loads to several processors. The main problem is to decide the optimal schedule of loads distribution to the processors to achieve the minimum solution time. Many DLT applications allow users to model the parallel system with linear equations or recursion, which can be solved efficiently.\par
Analysis in DLT was first studied by Cheng and Robertazzi in [4], which was designed originally for intelligent sensor networks. The formal proof of the DLT optimality principle was in [5], where a linear daisy chain network was applied. Since then, DLT has been well established and used in many scheduling problems. An analytic proof for a bus network that all processors must stop computing at the same time to obtain a minimal time solution was provided in [6]. In [7], optimal load distribution sequences for tree networks were investigated and in [8] computing cost was considered along with job finishing time. Closed-form expression for the processing time in the nonblocking mode of communication was derived in [9].  Scheduling divisible loads in a single-level tree network was considered in [9-14]. An optimal time-varying load scheduling for divisible loads was studied in [15], where the computing system was modeled as a bus-oriented network.\par
Most previous works assume that the channel speed and processing speed are constant throughout the whole processing time. It is often assumed that one processor can only process a single job at a time, which may not be true since in most practical computer systems one processor can both communicate with multiple networks and process multiple jobs. Such multi-task processors are commonly encountered in resource-sharing systems such as virtualized networks. In Wireless Sensor Networks (WSN) the same piece of WSN's physical resources can be virtualized into logical units, which can be used by multiple users [24]. Also, in the network-slicing technology for 5G networks, resource sharing among slices is sometimes permitted in order to maintain certain performance levels [25]. As a result, in such resource sharing systems those extra connections and jobs will take up the system resources and hence hinder the system processing of a specific job of our interest. In other word, the system speed may be time-varying according to the number of those extra loads. In this paper, the extra jobs running on a certain processor in addition to the job of our interest are called background jobs. For the area of time-varying scheduling studies, a data gathering problem is discussed in [23] where only data transmission is considered and the communication speed is time-varying. 
\subsection{Our Contribution}
The first work examining time-varying DLT is in [15] by Sohn and Robertazzi, where the loads are distributed through a bus network. The control processor, however, does not process data. In [15], the arrivals and departures times of background jobs are assumed to be exactly same for every processor, which is usually not true in practical situations. In our paper, each processor has its own background arrival and departure sequence, which is independent from others. Also, the processor sharing rule is updated in our paper. Instead as the processor (channel) devoting all its computational (transmission) power evenly to each job in [15], we assume that the processor (the channel) can assign an arbitrary ratio of its computational power (transmission power) to each job, as long as the sum of these ratio does not exceed one. Such an assumption is more realistic since modern virtualization technique allows users to divide the processor's computational (transmission) power according to their preference when a single physical processor is virtualized into multiple virtual processors. \par
Furthermore, a single level tree network with heterogeneous channels is used instead of the bus network at [15]. The single level tree network can model a variety of parallel systems using master-slave, or controller-worker paradigm. For instance, [16] models the case where several computers interconnected with an Ethernet as a single level tree network. Moreover, in [13] a single level tree network can be modeled as a set of computing clusters connected to a master controller via Internet. Moreover, this paper provides two algorithms for the stochastic analysis, which delivers superior performance compared to the one in [15].\par 
Also, in this paper, unlike [15], the control processor is equipped with a front-end sub-processor, which means it not only transfers data to other processors, but also processes data as well.\par 
Our objective is to determine the optimal partitions of the full load for each processor to achieve the minimum finishing time (makespan). Two cases are discussed in our paper: whether the control processor is a time-invariant processor or a time-varying processor, where the former one is a special case and latter one is more general. We first studied the deterministic model where the arrival and departure time points for the background jobs and extra connections are exactly known a priori. Two algorithms are provided for the two cases to solve the scheduling problem. Then a stochastic analysis is performed when those time points are not known a priori.\par
\subsection{Organization}
The rest of this paper is organized as follows. Section 2 first briefly introduces the classic solution of DLT scheduling problem in a time-invariant single level tree network. Then two time-varying cases are studied, respectively. The stochastic model is studied in section 3 and section 4 provides verification and evaluation of our method via different criterion. The conclusion appears in section 5.\par
The following notations are used in this paper:
\begin{enumerate} [leftmargin=1em,align=left]
	\item[$\alpha_{i}$] The partition of the entire divisible load that is assigned to processor $i$.
	\item[$W_{i}$] Inverse of processing speed of $ith$ processor when there is only one job.
	\item[$W_{i}(t)$] Inverse of time-varying processing speed of $ith$ processor applied to the divisible job at interest.
	\item[$\bar{W}_{i}$] Equivalent constant value of $W_{i}(t)$ during the processing time.
	\item[$T_{cp}$] Time to process the entire load when $W_{i} = 1$ for the $ith$ processor.
	\item[$Z_i$] Inverse of channel speed when control processor is only communicating with $ith$ processor.
	\item[$Z_{i}(t)$]  Inverse of time-varying channel speed applied to the divisible job at interest.
	\item[$\bar{Z}_{i}$] Equivalent constant value of $Z_{i}(t)$ when control processor is communicating with $ith$ worker processor
	\item[$Exp(\lambda)$] Negative exponential distribution with parameter $\lambda$.
	\item[$Unif(a,b)$] Uniform distribution with parameter $a,b$.
	\item[$T_{cm}$] Time to transmit the entire load when $Z = 1$.
	\item[$T_{f}$] The finishing time of processing the entire load.
\end{enumerate}
\section{Deterministic Analysis}
In this section, we assume that the arrival and departure times of background jobs are exactly known, which is referred as the deterministic model. We study the optimal scheduling for a time-varying single level tree system. To this end, we first briefly introduce the classic time-invariant problem, which will be helpful to understand the time-varying problem. The case that the exact arrival and departure times of background jobs are not known will be studied in the next section.
\subsection{Preliminaries}
Let's consider the single level tree network in Fig. \ref{fig:singleleveltree}. Assume that there are totally $N+1$ processors for the whole system.
%\FloatBarrier
\begin{figure}[h!]
	\centering
	\begin{tikzpicture}
	[processor/.style = {circle,draw,inner sep=0pt,minimum size=6mm}]
	\node[processor] (p0) at (0,1.5) {$P_{0}$};
	\node[processor] (p1) at (-1.5,-0.5) {$P_{1}$};
	\node[processor] (p2) at (-0.2,-0.5) {$P_{2}$};
	\node[processor] (p3) at (1.5,-0.5) {$P_{N}$};
	\draw [->] (p0) to (p1);
	\draw [->] (p0) to (p2);
	\draw [->] (p0) to (p3);
	\filldraw [black] (0.5,-0.5) circle [radius=0.5pt];
	\filldraw [black] (0.6,-0.5) circle [radius=0.5pt];
	\filldraw [black] (0.7,-0.5) circle [radius=0.5pt];
	\end{tikzpicture}
	\caption{Single level tree network} 
	\label{fig:singleleveltree}
\end{figure}
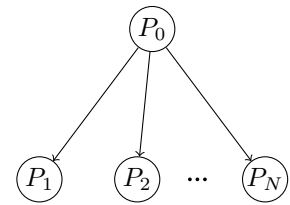
%\FloatBarrier
The processor $P_{0}$ is the control processor where the divisible load first arrived at. The control processor $P_{0}$ divides the divisible load to $N$ parts which is indicated by $\alpha_{1}, \alpha_{2},....,\alpha_{N}$ and assigns those $N$ parts to the worker processors $P_{1},P_{2},....,P_{N}$. In this paper we normalize the total amount of loads to be $1$, which means that $\alpha_{1} + \alpha_{2} + .... + \alpha_{N} = 1$. The worker processors are numbered in the order of receiving the loads. There are several assumptions for the processor:
\begin{itemize}
	\item A processor can only compute after it has finished the communication unless it is equipped with a front-end processor.
	\item The control processor can only communicate with one worker processor at a time (sequential load distribution).
	\item There is no communication between the worker processors.
\end{itemize}
In this case, we assume that due to a limitation of resources, only the control processor has a front-end processor, which means that it can compute at the same time as it communicates with other worker processors. According to the notation we define in section 1, the piece of load $\alpha_{i}$ is transferred to worker processor $P_{i}$ in time $\alpha_{i}ZT_{cm}$ and is processed in time $\alpha_{i}W_{i}T_{cp}$. All the processors should finish computing at the same moment to achieve the smallest $T_{f}$ by the optimality principle proved in [4,13,17-20]. Our problem is to find the load partitions $\alpha_{1}, \alpha_{2},....,\alpha_{N}$ when the optimality principle is achieved.\par
We can draw the timing diagram according to those conditions in Fig. \ref{fig:timeinvariant}.
%\FloatBarrier
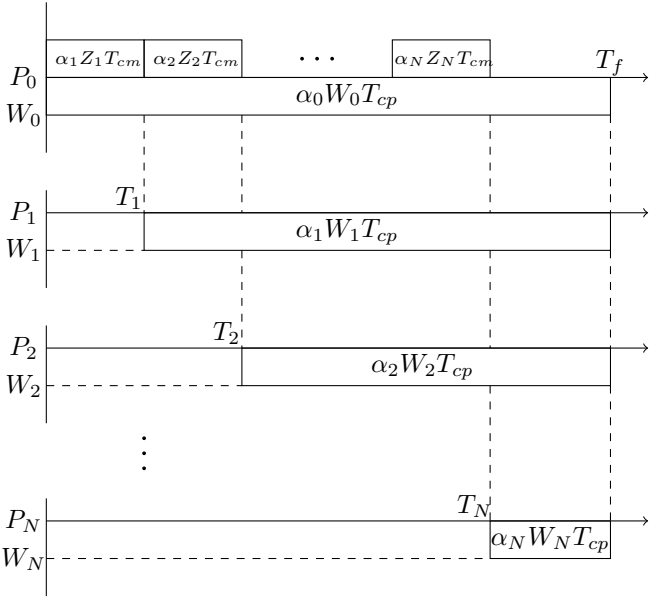
\begin{figure}[h!]
	\centering
	\begin{tikzpicture}
	%P0
	%time lane
	\draw[->] (-4.5,3) -- (3.5,3);
	%start lane
	\draw (-4.5,2) -- (-4.5,4);
	\draw (-4.8,3) node {$P_{0}$};
	%Communication blocks
	\draw (-4.5,3) rectangle (-3.2,3.5);
	%\draw (-4.8,3.5) node {$Z$};
	\draw (-3.8,3.25) node {\scriptsize $\alpha_{1}Z_{1}T_{cm}$};
	\draw (-3.2,3) rectangle (-1.9,3.5);
	\draw (-2.5,3.25) node {\scriptsize $\alpha_{2}Z_{2}T_{cm}$};
	%	\draw (-1.9,3) rectangle (-0.6,3.5);
	%	\draw (-1.2,3.25) node {$\alpha_{3}ZT_{cm}$};
	\draw (0.1,3) rectangle (1.4,3.5);
	\draw (0.8,3.25) node {\scriptsize $\alpha_{N}Z_{N}T_{cm}$};
	\filldraw [black] (-1.1,3.25) circle [radius=0.5pt];
	\filldraw [black] (-0.9,3.25) circle [radius=0.5pt];
	\filldraw [black] (-0.7,3.25) circle [radius=0.5pt];
	%Computing blocks
	\draw (-4.5,3) rectangle (3,2.5);
	\draw (-4.8,2.5) node {$W_{0}$};
	\draw (-0.5,2.75) node {$\alpha_{0}W_{0}T_{cp}$};
	%\draw (4.5,3) node {\verb|Time|};
	%P1
	\draw[->] (-4.5,1.2) -- (3.5,1.2);
	\draw (-4.5,0.2) -- (-4.5,1.5);
	\draw (-4.8,1.2) node {$P_{1}$};
	\draw (-3.2,1.2) rectangle (3,0.7);
	\draw[dashed] (-4.5,0.7) -- (-3.2,0.7);
	\draw (-4.8,0.7) node {$W_{1}$};
	\draw (-0.5,0.95) node {$\alpha_{1}W_{1}T_{cp}$};
	%\draw (4.5,1.2) node {\verb|Time|};
	\draw[dashed] (-3.2,2.5) -- (-3.2,1.2);
	\draw (-3.4,1.4) node {$T_{1}$};
	%P2
	\draw[->] (-4.5,-0.6) -- (3.5,-0.6);
	\draw (-4.5,-1.6) -- (-4.5,-0.3);
	\draw (-4.8,-0.6) node {$P_{2}$};
	\draw (-1.9,-0.6) rectangle (3,-1.1);
	\draw[dashed] (-4.5,-1.1) -- (-1.9,-1.1);
	\draw (-4.8,-1.1) node {$W_{2}$};
	\draw (0.5,-0.85) node {$\alpha_{2}W_{2}T_{cp}$};
	%\draw (4.5,-0.6) node {\verb|Time|};
	\draw[dashed] (-1.9,2.5) -- (-1.9,1.2);
	\draw[dashed] (-1.9,0.7) -- (-1.9,-0.6);	
	\draw (-2.1,-0.4)node {$T_{2}$};
	\filldraw [black] (-3.2,-1.8) circle [radius=0.5pt];
	\filldraw [black] (-3.2,-2.0) circle [radius=0.5pt];
	\filldraw [black] (-3.2,-2.2) circle [radius=0.5pt];
	%Pn
	\draw[->] (-4.5,-2.9) -- (3.5,-2.9);
	\draw (-4.5,-3.9) -- (-4.5,-2.6);
	\draw (-4.8,-2.9) node {$P_{N}$};
	\draw (1.4,-2.9) rectangle (3,-3.4);
	\draw[dashed] (-4.5,-3.4) -- (1.4,-3.4);
	\draw (-4.8,-3.4)node {$W_{N}$};
	\draw (2.2,-3.15) node {$\alpha_{N}W_{N}T_{cp}$};
	%\draw (4.5,-2.9) node {\verb|Time|};
	\draw[dashed] (1.4,2.5) -- (1.4,1.2);
	\draw[dashed] (1.4,0.7) -- (1.4,-0.6);
	\draw[dashed] (1.4,-1.1) -- (1.4,-2.9);
	\draw(1.2,-2.7)node {$T_{N}$};
	\draw[dashed] (3,2.5) -- (3,1.2);
	\draw[dashed] (3,0.7) -- (3,-0.6);
	\draw[dashed] (3,-1.1) -- (3,-2.9);
	\draw(3,3.2)node {$T_{f}$};
	\end{tikzpicture}
	\caption{Timing diagram for single level tree network} 
	\label{fig:timeinvariant}
\end{figure}
%\FloatBarrier
For each time axis in the timing diagram, communication appears above the axis while computing appears below the axis. At $t=0$, the control processor starts sending partition $\alpha_{1}$ to worker $P_{1}$ in time $\alpha_{1}Z_{1}T_{cm}$. At $t=T_{1}$, after receiving the loads, $P_{1}$ starts processing and finishes in time $\alpha_{1}W_{1}T_{cp}$. This procedure repeats for every worker processor and all the processors finish computing at the same time $t = T_{f}$. The linear system equations can then be expressed as:\\
\begin{subequations}
	\begin{align}
	T_{f} = \alpha_{0}W_{0}T_{cp}\\
	T_{f} = \sum_{k = 1}^{i} \alpha_{k}Z_{k}T_{cm} + \alpha_{i}W_{i}T_{cp}, \hspace{0.2cm} i = 1,2,...N\\
	\sum_{k = 1}^{N}\alpha_{k} = 1
	\end{align}
\end{subequations}
Since there are $N+2$ unknowns and $N+2$ linear equations, load partitions $\alpha_{1}, \alpha_{2},....,\alpha_{N}$ can be uniquely solved as well as the $T_{f}$.
\subsection{Time-varying System with A Time-invariant Control Processor}
Now we consider that the processors can simultaneously process multiple jobs, which means in addition to the divisible job we studied in section 2.1, the processor also processes some other jobs. We call those jobs as background jobs. The background jobs will take the computing power from the processor and as a result the processor's processing speed will vary according to the amount of workload over time. In this section, we only consider that the worker processors are time-varying. The time-varying control processor will be discussed in the next section.\par
\begin{figure}[h!]
	\centering
	\includegraphics[width=0.5\textwidth]{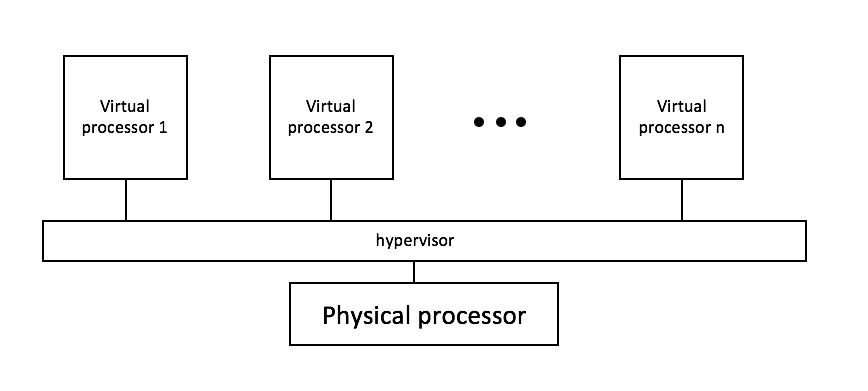}
	\caption{Processor virtualization}
	\label{fig:vir}
\end{figure}	
When a processor processes multiple jobs in parallel, the processor is virtualized into multiple virtual processors. In this way, each user of the system feels that it is the exclusive user of the processor. As shown in Fig. \ref{fig:vir}, there is a hypervisor controlling the virtualization process. The hypervisor can assign any ratio of computation power to any virtual processor, as long as the sum computation power of all virtual processors does not exceed the physical processor's maximum computation power. The protocol for the hypervisor to assign the physical processor's computation power is pre-defined in the hypervisor. As a result, the processing speed for the divisible load job of our interest is a function of the number of jobs in the processor defined by the hypervisor. For the case that $n$ jobs in the processor $i$, we use $W_{i}^{h}(n)$ to denote the inverse of computing speed applied to the divisible job of interest. This $W_{i}^{h}(n)$ is supposed to be known once $n$ is given. If $n=1$ and there is only the divisible load job in the processor, we denote the processing speed as $W_{i}$ for simplicity. We also use $W_i(t)$ to represent the general time-varying inverse of computing speed applied to the divisible job at interest and $\bar{W}_i$ to represent the equivalent constant value of $W_i(t)$ during processing for processor $i$. The background jobs arrive and leave independently on different processor. The method to define $\bar{W}_i$ will be introduced in this section. By adapting the time-varying processing speed to Fig. \ref{fig:timeinvariant}, the timing diagram for this condition can be depicted in Fig. \ref{fig:case1time}. \par
%\FloatBarrier
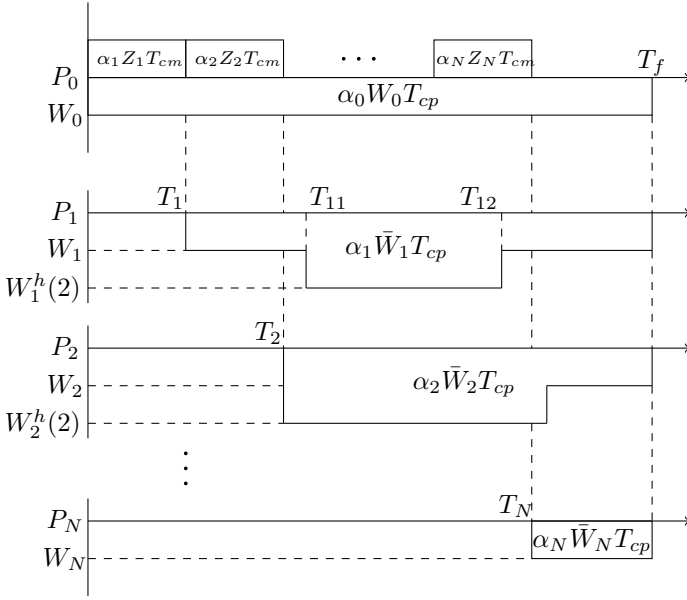
\begin{figure}[h!]
	\centering
	\begin{tikzpicture}
	%P0
	%time lane
	\draw[->] (-4.5,3) -- (3.5,3);
	%start lane
	\draw (-4.5,2) -- (-4.5,4);
	\draw (-4.8,3) node {$P_{0}$};
	%Communication blocks
	\draw (-4.5,3) rectangle (-3.2,3.5);
	%\draw (-4.8,3.5) node {$Z$};
	\draw (-3.8,3.25) node {\scriptsize $\alpha_{1}Z_{1}T_{cm}$};
	\draw (-3.2,3) rectangle (-1.9,3.5);
	\draw (-2.5,3.25) node {\scriptsize $\alpha_{2}Z_{2}T_{cm}$};
	\draw (0.1,3) rectangle (1.4,3.5);
	\draw (0.8,3.25) node {\scriptsize $\alpha_{N}Z_{N}T_{cm}$};
	\filldraw [black] (-1.1,3.25) circle [radius=0.5pt];
	\filldraw [black] (-0.9,3.25) circle [radius=0.5pt];
	\filldraw [black] (-0.7,3.25) circle [radius=0.5pt];
	%Computing blocks
	\draw (-4.5,3) rectangle (3,2.5);
	\draw (-0.5,2.75) node {$\alpha_{0}W_{0}T_{cp}$};
	\draw (-4.8,2.5) node {$W_{0}$};
	%\draw (4.5,3) node {\verb|Time|};
	%P1
	\draw[->] (-4.5,1.2) -- (3.5,1.2);
	\draw (-4.5,0) -- (-4.5,1.5);
	\draw (-4.8,1.2) node {$P_{1}$};
	%\draw (-3.2,1.2) rectangle (3.7,0.7);
	\draw (-3.2,1.2) -- (-3.2,0.7);
	\draw (-3.2,0.7) -- (-1.6,0.7);
	\draw (-1.6,0.7) -- (-1.6,0.2);
	\draw (-1.6,0.2) -- (1,0.2);
	\draw (1,0.2) -- (1,0.7);
	\draw (1,0.7) -- (3,0.7);
	\draw (3,0.7) -- (3,1.2);
	\draw (-0.4,0.75) node {$\alpha_{1}\bar{W}_{1}T_{cp}$};
	%\draw (4.5,1.2) node {\verb|Time|};
	\draw[dashed] (-3.2,2.5) -- (-3.2,1.2);
	\draw (-3.4,1.4) node {$T_{1}$};
	\draw[dashed] (-1.6,1.2) -- (-1.6,0.7);
	\draw (-1.3,1.4) node {$T_{11}$};
	\draw[dashed](1,1.2) -- (1,0.7);
	\draw (0.7,1.4) node {$T_{12}$};
	\draw[dashed] (-4.5,0.7) -- (-3.2,0.7);
	\draw (-4.8,0.7) node {$W_{1}$};
	\draw[dashed] (-4.5,0.2) -- (-1.6,0.2);
	\draw (-5.1,0.2) node {$W_{1}^{h}(2)$};
	%P2
	\draw[->] (-4.5,-0.6) -- (3.5,-0.6);
	\draw (-4.5,-1.8) -- (-4.5,-0.3);
	\draw (-4.8,-0.6) node {$P_{2}$};
	%\draw (-1.9,-0.6) rectangle (3.7,-1.1);
	\draw (-1.9,-0.6) -- (-1.9,-1.6);
	\draw (-1.9,-1.6) -- (1.6,-1.6);
	\draw (1.6,-1.6) -- (1.6,-1.1);
	\draw (1.6,-1.1) -- (3,-1.1);
	\draw (3,-1.1) -- (3,-0.6);
	\draw (0.5,-1.05) node {$\alpha_{2}\bar{W}_{2}T_{cp}$};
	\draw[dashed] (-4.5,-1.1) -- (-1.9,-1.1);
	\draw (-4.8,-1.1) node {$W_{2}$};
	\draw[dashed] (-4.5,-1.6) -- (-1.9,-1.6);
	\draw (-5.1,-1.6) node {$W_{2}^{h}(2)$};
	%\draw (4.5,-0.6) node {\verb|Time|};
	\draw[dashed] (-1.9,2.5) -- (-1.9,1.2);
	\draw[dashed] (-1.9,0.7) -- (-1.9,-0.6);	
	\draw (-2.1,-0.4)node {$T_{2}$};
	\filldraw [black] (-3.2,-2.0) circle [radius=0.5pt];
	\filldraw [black] (-3.2,-2.2) circle [radius=0.5pt];
	\filldraw [black] (-3.2,-2.4) circle [radius=0.5pt];
	%Pn
	\draw[->] (-4.5,-2.9) -- (3.5,-2.9);
	\draw (-4.5,-3.9) -- (-4.5,-2.6);
	\draw (-4.8,-2.9) node {$P_{N}$};
	\draw (1.4,-2.9) rectangle (3,-3.4);
	\draw (2.2,-3.15) node {$\alpha_{N}\bar{W}_{N}T_{cp}$};
	\draw[dashed] (-4.5,-3.4) -- (1.4,-3.4);
	\draw (-4.8,-3.4) node {$W_{N}$};
	\draw[dashed] (1.4,2.5) -- (1.4,1.2);
	\draw[dashed] (1.4,0.7) -- (1.4,-0.6);
	\draw[dashed] (1.4,-1.6) -- (1.4,-2.9);
	\draw(1.2,-2.7)node {$T_{N}$};
	\draw[dashed] (3,2.5) -- (3,1.2);
	\draw[dashed] (3,0.7) -- (3,-0.6);
	\draw[dashed] (3,-1.1) -- (3,-2.9);
	\draw(3,3.2)node {$T_{f}$};
	\end{tikzpicture}
	\caption{Timing diagram for single level tree network with time-varying worker processor speed} 
	\label{fig:case1time}
\end{figure}
%\FloatBarrier
In Fig. \ref{fig:case1time}, we use steps to represent the arrival and departure of the background jobs, and the value of $W(t)$ is noted on the vertical axis. For example, processor $P_1$ starts processing the data at time $T_{1}$, then one background job arrives at time $T_{11}$ where a down step appears. Note that $W$ is the inverse of processing speed, and processing speed jumps from $\frac{1}{W_{1}}$ to $\frac{1}{W_{1}^{h}(2)}$ at this time point, thus $W_{1}(t)$ jumps from $W_{1}$ to $W_{1}^{h}(2)$. Afterwards, this background job departs at time $T_{12}$ where an up step appears. $W_{1}(t)$ also jumps back from $W_{1}^{h}(2)$ to $W_{1}$. In this section we assume that the time points of arrival and departure of the background jobs are exactly known, which mean that $W_{i}(t), i=1,2,...,N$ are exactly known. \par
Theorem I shows how to achieve the $\bar{W}_{i}, i = 1,2,...,N$. 
\begin{theorem}
	The equivalent constant value of $W_{i}(t)$ during the processing of $ith$ processor equals to:
	\[\bar{W}_{i} = \frac{T_{f} - T_{i}}{\int_{T_{i}}^{T_{f}} \frac{1}{W_{i}(t)}dt}\]
	where the $T_{i}$ denotes the start time of $ith$ processor's computation and $T_{f}$ denotes the finishing time.
\end{theorem}
\begin{proof}
	Since the changes in $W_{i}(t)$ are all steps at certain time points where background jobs arrive and depart, let's assume that for $ith$ processor there are $k$ changes in $W_{i}(t)$ between time $T_{i}$ to $T_{f}$ and let $T_{ij}, j = 1,2,3...k$ denote the $jth$ change time point. For example in Fig. \ref{fig:case1time}, $T_{11}$ is the first time point of the change in $W_{1}(t)$ after $T_{1}$ and $T_{12}$ is the last time point of change in $W_{1}(t)$ before $T_{f}$. Let $W_{i(j+1)}$ denotes the value of $W_{i}(t)$ between $T_{ij}$ to $T_{i(j+1)}$, where $j = 1,2,...k-1$ and $W_{i1}$ between time $T_{i}$ to $T_{i1}$, $W_{i(k+1)}$ between time $T_{ik}$ to $T_{f}$. Also, in the same manner, let $\alpha_{i(j+1)}$ denotes the partition of loads that processed between time $T_{ij}$ to $T_{i(j+1)}$, where $j = 1,2,...k-1$ and $\alpha_{i1}$ between time $T_{i}$ to $T_{i1}$, $\alpha_{i(k+1)}$ between time $T_{ik}$ to $T_{f}$. Then we have the equations:
	\begin{subequations}
		\begin{align}
		T_{i1} - T_{i} = \alpha_{i1}W_{i1}T_{cp}\\
		T_{i(j+1)} - T_{ij} = \alpha_{i(j+1)}W_{i(j+1)}T_{cp}, j=1,2,..,k-1\\
		T_{f} - T_{ik} = \alpha_{i(k+1)} W_{i(k+1)}T_{cp}\\
		\alpha_{i} = \sum_{j=1}^{k+1}\alpha_{ij}
		\end{align}
	\end{subequations}
	Since by definition $T_{f} - T_{i} = \alpha_{i}\bar{W}_{i}T_{cp}$, then:
	\begin{align}
	\bar{W}_{i} = \frac{T_{f} - T_{i}}{\alpha_{i}T_{cp}}
	\end{align}
	By substituting equation (2) into equation (3), we can get:
	\begin{align*}
	\bar{W}_{i} &= \frac{T_{f} - T_{i}}{\alpha_{i}T_{cp}}\\
	&= \frac{T_{f} - T_{i}}{(\sum_{j=1}^{k+1}\alpha_{ij})T_{cp}}\\
	&= \frac{T_{f} - T_{i}}{\frac{T_{i1} - T_{i}}{W_{i1}} + \sum_{j=1}^{k-1}\frac{T_{i(j+1)} - T_{ij}}{W_{i(j+1)}} + \frac{T_{f} - T_{ik}}{W_{i(k+1)}}}\\
	&= \frac{T_{f} - T_{i}}{\int_{T_{i}}^{T_{f}} \frac{1}{W_{i}(t)}dt}
	\end{align*}
	This completes the proof of Theorem I.
\end{proof}
\begin{remark}
	The inverse of $\bar{W}_{i}$ equals to $\frac{\int_{T_{i}}^{T_{f}} \frac{1}{W_{i}(t)}dt}{T_{f} - T_{i}}$, which is the average value of $\frac{1}{W_{i}(t)}$ between $T_{i}$ to $T_{f}$. Since $W_{i}(t)$ is defined as the inverse of computation speed of the $ith$ processor, $\bar{W}_{i}$ can also be taken as the inverse of the average computing speed, which is the inverse of the average value of $\frac{1}{W_{i}(t)}$.
\end{remark}
Based on the expression of $\bar{W}_{i}$, the system equations can be written as:\\
\begin{subequations}
	\begin{align}
	\alpha_{0}W_{0}T_{cp} = T_{f}\\
	T_{i} = \sum_{k=1}^{i} \alpha_{k}Z_{k}T_{cm}, i = 1,2,...,N\\
	T_{f} - T_{i} = \alpha_{i}\bar{W}_{i}T_{cp}, i = 1,2,...,N\\
	\alpha_{0} + \alpha_{1} + ... + \alpha_{N} = 1
	\end{align}
\end{subequations}
where equation (4b) represents the communication time for each processor and equation (4a) and (4c) represent the computation time. Equation (4d) guarantees that all the partitions sum up to 1. From equation (4b), we can express $\alpha_{i}$ as a function of $T_{i-1}$ and $T_{i}$ as $\alpha_{i} = \frac{T_{i} - T_{i-1}}{Z_{i}T_{cm}}$. By substituting this transformation into equation (4c) we have:\\
\begin{subequations}
	\begin{align}
	T_f &= T_i + \frac{T_{i} - T_{i-1}}{Z_{i}T_{cm}}\bar{W}_{i}T_{cp}\\
	&= T_i + \frac{T_{i} - T_{i-1}}{Z_{i}T_{cm}}\frac{T_{f} - T_{i}}{\int_{T_{i}}^{T_{f}} \frac{1}{W_{i}(t)}dt}T_{cp}
	\end{align}
\end{subequations}
Starting from processor $1$, equation (5b) can be reduced as $T_f = T_1 + \frac{T_{1}}{Z_{1}T_{cm}}\frac{T_{f} - T_{1}}{\int_{T_{1}}^{T_{f}} \frac{1}{W_{1}(t)}dt}T_{cp}$, which is an equation of $T_{1}$ and $T_{f}$ only. Thus $T_{1}$ can be expressed as a function of $T_f$ only. By the definition $\alpha_1$ can also be expressed as a function of $T_f$ only. This provide an intuition that this problem can be solved recursively. A recursive algorithm is introduced to calculate the optimal finishing time $T_{f}$ and partitions $\alpha_{i}$ as Algorithm I.
\begin{algorithm}
	1. Express $\alpha_{0}$ as a function of $T_{f}$ using the equation:
	\begin{align*}
	\alpha_{0} = \frac{T_{f}}{W_{0}T_{cp}}
	\end{align*}
	\\
	2. Express $T_{1}$ as a function of $T_{f}$ using the equation:
	\begin{align*}
	T_f = T_1 + \frac{T_{1}}{Z_{1}T_{cm}}\frac{T_{f} - T_{1}}{\int_{T_{1}}^{T_{f}} \frac{1}{W_{1}(t)}dt}T_{cp}
	\end{align*}
	%where $\alpha_{0}$ is a function of $T_{f}$\\
	Express $\alpha_{1}$ as a function of $T_{f}$ using the equation:
	\begin{align*}
	\alpha_{1} = \frac{T_{1}}{Z_{1}T_{cm}}
	\end{align*}
	%where $T_{N}$ is a function of $T_{f}$ \\
	3. Express $T_{2}$ as a function of $T_{f}$ using the equation:
	\begin{align*}
	T_f =T_2 + \frac{T_{2} - T_{1}}{Z_{2}T_{cm}}\frac{T_{f} - T_{2}}{\int_{T_{2}}^{T_{f}} \frac{1}{W_{2}(t)}dt}T_{cp}
	\end{align*}
	where $T_{1}$ is a function of $T_{f}$\\
	Express $\alpha_{2}$ as a function of $T_{f}$ using the equation:
	\begin{align*}
	\alpha_{2} = \frac{T_{2} - T_{1}}{Z_{2}T_{cm}}
	\end{align*}
	where $T_{1}$ and $T_{2}$ are functions of $T_{f}$\\
	4.Repeat the procedure until $\alpha_{N}$ is expressed as a function of $T_{f}$. Now, every $\alpha_{i}$ has been expressed as a function of $T_{f}$. \\
	5.Apply the normalization equation:
	\begin{align*}
	\alpha_{0} + \alpha_{1} + ... + \alpha_{N} = 1
	\end{align*}
	to calculate the optimal finishing time $T_{f}$, as well as all the partitions $\alpha_{i}$s.
	\caption{Recursive algorithm to solve the optimal scheduling problem in a time-varying system with a time-invariant control processor}
\end{algorithm}
\subsection{Time-varying Control Processor, Processing and Communication Speed}
In the previous section, we studied the optimal scheduling problem for a single level tree network where the worker processors have time-varying processing speeds due to the arrival and departure of background jobs. In this section, we consider the general case that the background jobs appear on the control processor as well, which will make the processing speed time-varying for $P_{0}$. Also, we assume that there will be other transmissions such as the control processor communicating with other networks when assigning the loads, which will slow down the communication speed for the job of our interest. This will make the communication speed time-varying. Similar as the previous subsection a processor is virtualized into multiple virtual processors to share the communication power and there is a hypervisor to control them. Same as $W_i^h(n)$ and $W_i(t)$, we use $Z_i^h(n)$ and $Z_i(t)$ to represent the time-varying inverse of communication speed applied to the divisible job at interest. As a result, $Z_i(t)$ will also be a function of steps. Again, we assume that the time points when links established and finished with other networks are known for each processor, which means $Z_i(t)$ is exactly known. For simplicity we use $Z_i$ to represent the inverse of communication speed when there is only the divisible load job of our interest in the control processor for distribution.\par 
%\FloatBarrier
\begin{figure}[h!]
	\centering
	\begin{tikzpicture}
	%P0
	%time lane
	\draw[->] (-4.5,3) -- (3.5,3);
	%start lane
	\draw (-4.5,1.8) -- (-4.5,4.4);
	\draw (-4.8,3) node {$P_{0}$};
	%Communication blocks
	%\draw (-4.5,3) rectangle (-3.2,3.5);
	\draw (-4.5,3.5) -- (-4,3.5);
	\draw (-4,3.5) -- (-4,3.8);
	\draw (-4,3.8) -- (-2,3.8);
	\draw (-2,3.8) -- (-2,4.1);
	\draw (-2,4.1) -- (-0.5,4.1);
	\draw (-0.5,4.1) -- (-0.5,3.8);
	\draw (-0.5,3.8) -- (1.4,3.8);
	\draw (1.4,3.8) -- (1.4,3);
	\draw (-2.8,3.8) -- (-2.8,3);%T1
	\draw (-3.6,3.25) node {$\alpha_{1}\bar{Z_{1}}T_{cm}$};
	\draw (-1.1,4.1) -- (-1.1,3);%T2
	\draw (-1.9,3.25) node {$\alpha_{2}\bar{Z}_{2}T_{cm}$};
	\draw (-0.3,3.8) -- (-0.3,3);%TN
	\draw (0.6,3.25) node {$\alpha_{N}\bar{Z}_{N}T_{cm}$};
	\draw[dashed] (-4.5,3.8) -- (-4,3.8);
	\draw[dashed] (-4.5,4.1) -- (-2,4.1);
	\draw (-4.8,3.5) node {$Z_1$};
	\draw (-5.0,3.9) node {$Z_1^h(2)$};
	\draw (-5.0,4.3) node {$Z_2^h(3)$};
	\filldraw [black] (-0.8,3.25) circle [radius=0.5pt];
	\filldraw [black] (-0.7,3.25) circle [radius=0.5pt];
	\filldraw [black] (-0.6,3.25) circle [radius=0.5pt];
	\draw (-4.5,2.5) -- (-3.5,2.5);
	\draw (-3.5,2.5) -- (-3.5,2);
	\draw (-3.5,2) -- (-2.6,2);
	\draw (-2.6,2) -- (-2.6,2.5);
	\draw (-2.6,2.5) -- (-2,2.5);
	\draw (-2,2.5) -- (-2,2);
	\draw (-2,2) -- (1,2);
	\draw (1,2) -- (1,2.5);
	\draw (1,2.5) -- (2,2.5);
	\draw (2,2.5) -- (2,2);
	\draw (2,2) -- (3,2);
	\draw (3,2) -- (3,3);
	\draw[dashed] (-4.5,2) -- (-3.5,2);
	\draw (-4.8,2.5) node {$W_{0}$};
	\draw (-5.0,2) node {$W_{0}^h(2)$};
	\draw (-0.5,2.55) node {$\alpha_{0}\bar{W}_{0}T_{cp}$};
	\draw[->] (-4.5,1.2) -- (3.5,1.2);
	\draw (-4.5,0) -- (-4.5,1.5);
	\draw (-4.8,1.2) node {$P_{1}$};
	\draw (-2.8,1.2) -- (-2.8,0.7);
	\draw (-2.8,0.7) -- (-1.6,0.7);
	\draw (-1.6,0.7) -- (-1.6,0.2);
	\draw (-1.6,0.2) -- (1,0.2);
	\draw (1,0.2) -- (1,0.7);
	\draw (1,0.7) -- (3,0.7);
	\draw (3,0.7) -- (3,1.2);
	\draw[dashed] (-4.5,0.7) -- (-2.8,0.7);
	\draw[dashed] (-4.5,0.2) -- (-1.6,0.2);
	\draw (-4.8,0.7) node {$W_{1}$};
	\draw (-5.0,0.2) node {$W_{1}^h(2)$};
	\draw (-0.4,0.75) node {$\alpha_{1}\bar{W}_{1}T_{cp}$};
	\draw[dashed] (-2.8,3) -- (-2.8,1.2);
	\draw (-3,1.4) node {$T_{1}$};
	\draw[->] (-4.5,-0.6) -- (3.5,-0.6);
	\draw (-4.5,-1.8) -- (-4.5,-0.3);
	\draw (-4.8,-0.6) node {$P_{2}$};
	\draw (-1.1,-0.6) -- (-1.1,-1.6);
	\draw (-1.1,-1.6) -- (1.6,-1.6);
	\draw (1.6,-1.6) -- (1.6,-1.1);
	\draw (1.6,-1.1) -- (3,-1.1);
	\draw (3,-1.1) -- (3,-0.6);
	\draw[dashed] (-4.5,-1.6) -- (-1.1,-1.6);
	\draw[dashed] (-4.5,-1.1) -- (-1.1,-1.1);
	\draw (-4.8,-1.1) node {$W_{2}$};
	\draw (-5.0,-1.6) node {$W_{2}^h(2)$};
	\draw (0.5,-1.05) node {$\alpha_{2}\bar{W}_{2}T_{cp}$};
	\draw[dashed] (-1.1,3) -- (-1.1,1.2);
	\draw[dashed] (-1.1,0.2) -- (-1.1,-0.6);	
	\draw (-1.3,-0.4)node {$T_{2}$};
	\filldraw [black] (-3.2,-2.0) circle [radius=0.5pt];
	\filldraw [black] (-3.2,-2.2) circle [radius=0.5pt];
	\filldraw [black] (-3.2,-2.4) circle [radius=0.5pt];
	\draw[->] (-4.5,-2.9) -- (3.5,-2.9);
	\draw (-4.5,-3.9) -- (-4.5,-2.6);
	\draw (-4.8,-2.9) node {$P_{N}$};
	\draw (1.4,-2.9) rectangle (3,-3.4);
	\draw[dashed] (-4.5,-3.4) -- (1.4,-3.4);
	\draw (-4.8,-3.4) node {$W_{N}$};
	\draw (2.2,-3.15) node {$\alpha_{N}\bar{W}_{N}T_{cp}$};
	\draw[dashed] (1.4,2.5) -- (1.4,1.2);
	\draw[dashed] (1.4,0.7) -- (1.4,-0.6);
	\draw[dashed] (1.4,-1.6) -- (1.4,-2.9);
	
	\draw(1.2,-2.7)node {$T_{N}$};
	\draw[dashed] (3,2.5) -- (3,1.2);
	\draw[dashed] (3,0.7) -- (3,-0.6);
	\draw[dashed] (3,-1.1) -- (3,-2.9);
	\draw(3,3.2)node {$T_{f}$};
	\end{tikzpicture}
	\caption{Timing diagram for single level tree network with time-varying channel speed and computing speed } 
	\label{fig:case2time}
\end{figure}
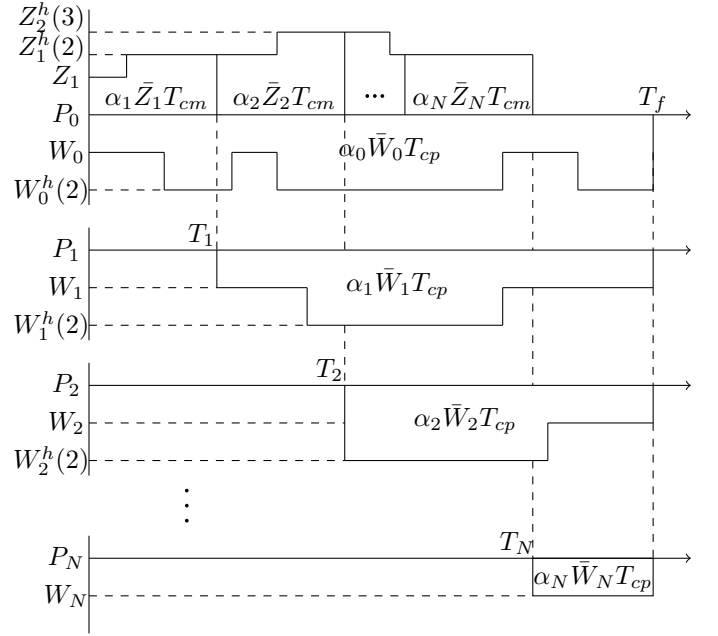
%\FloatBarrier
Fig. \ref{fig:case2time} demonstrates the timing diagram for a general time-varying single level tree system. The channel speed varies as well as the computing speed for each processor. At the beginning, $P_{0}$ starts to transmit partition $\alpha_{1}$ of the loads to $P_{1}$ and finishes at time $T_{1}$. After finishing receiving the loads, $P_{1}$ starts to process at the time point $T_{1}$, while $P_{0}$ starts to transmit the partition $\alpha_{2}$ to $P_{2}$. This procedure repeats for every processor, and again, every processor finishes at the same time $T_{f}$ for the optimal condition.\par
Apparently, Theorem I still works in this situation. Similarly, we can find the expression for the equivalent time-invariant value of $Z(t)$:
\begin{theorem}
	The equivalent constant value of $Z_i(t)$ when $P_{0}$ communicates with $P_{i}$ equals to:
	\[\bar{Z}_{i} = \frac{T_{i} - T_{i-1}}{\int_{T_{i-1}}^{T_{i}} \frac{1}{Z_i(t)}dt}\]
	where the $T_{i}$ denotes the $ith$ processor's start computing time, $i = 1,2,...,N$, $T_{0} = 0$. 
\end{theorem}
The proof should be similar to the proof of the Theorem 1. Also, the $\bar{Z}_{i}$ can be taken as the inverse of average communication speed, in the same manner as $\bar{W}_{i}$ in the remark of Theorem I. \par
\begin{subequations}
	\begin{align}
	\alpha_{0}\bar{W}_{0}T_{cp} = T_{f}\\
	T_{i} - T_{i-1} = \alpha_{i}\bar{Z}_{i}T_{cm}, i = 1,2,...,N\\
	T_{i} = \sum_{k=1}^{i}\alpha_{i}\bar{Z}_{i}T_{cm}, i = 1,2,...,N \\ 
	T_{f} - T_{i} = \alpha_{i}\bar{W}_{i}T_{cp}, i = 1,2,...,N\\
	\alpha_{0} + \alpha_{1} + ... + \alpha_{N} = 1
	\end{align}
\end{subequations}
Equations (6a) - (6d) demonstrate the system equations when $P_{0}$ is also time-varying. Similar to equations (4a) to (4d), equations (6a) and (6d) represent the processing part for each processor, (6b) and (6c) represent the communication part and (6e) represents the normalization equation. \\
By applying Theorem II to equation (6b):
\begin{subequations}
	\begin{align}
	T_{i} - T_{i-1} = \alpha_{i}\frac{T_{i} - T_{i-1}}{\int_{T_{i-1}}^{T_{i}} \frac{1}{Z_i(t)}dt} T_{cm}\\
	\Longrightarrow 
	\alpha_{i} = \frac{1}{T_{cm}}\int_{T_{i-1}}^{T_{i}} \frac{1}{Z_i(t)}dt
	\end{align}  
\end{subequations}
We can find that $\alpha_{i}$ is the integral of $\frac{1}{Z_i(t)}$ from $T_{i-1}$ to $T_{i}$ times a constant $T_{cm}$. Since $Z_i(t)$ is assumed to be known, by applying the same recursive method as last subsection, we can express every $\alpha$s as a function of $T_f$ and using the normalization equation to solve the optimal scheduling problem. The detailed steps are introduced in Algorithm II.
% where $\alpha_{0}$ can also be expressed as a function of $T_{f}$ only using equation (6a). An recursive algorithm is introduced below. 
\begin{algorithm}
	1. Express $\alpha_{0}$ as a function of $T_{f}$ using the equation:
	\begin{align*}
	\alpha_{0} = \frac{1}{T_{cp}}\int_{0}^{T_{f}} \frac{1}{W_{0}(t)}dt
	\end{align*}
	\\
	2. Express $T_{1}$ as a function of $T_{f}$ using the equation:
	\begin{align*}
	T_f = T_1 + \frac{1}{T_{cm}}\int_{0}^{T_{1}}\frac{1}{Z_1(t)}dt\frac{T_{f} - T_{1}}{\int_{T_{1}}^{T_{f}} \frac{1}{W_{1}(t)}dt}T_{cp}
	\end{align*}
	%where $\alpha_{0}$ is a function of $T_{f}$\\
	Express $\alpha_{1}$ as a function of $T_{f}$ using the equation:
	\begin{align*}
	\alpha_{1} = \frac{1}{T_{cm}}\int_{0}^{T_{1}}\frac{1}{Z_1(t)}dt
	\end{align*}
	where $T_{1}$ is a function of $T_{f}$ \\
	3. Express $T_{2}$ as a function of $T_{f}$ using the equation:
	\begin{align*}
	T_f = T_2 + \frac{1}{T_{cm}}\int_{T_{1}}^{T_{2}}\frac{1}{Z_2(t)}dt\frac{T_{f} - T_{2}}{\int_{T_{2}}^{T_{f}} \frac{1}{W_{2}(t)}dt}T_{cp}
	\end{align*}
	where $T_{1}$ is a function of $T_{f}$\\
	Express $\alpha_{2}$ as a function of $T_{f}$ using the equation:
	\begin{align*}
	\alpha_{2} = \frac{1}{T_{cm}}\int_{T_{1}}^{T_{2}}\frac{1}{Z_2(t)}dt
	\end{align*}
	where $T_{2}$ and $T_{1}$ are functions of $T_{f}$\\
	4.Repeat the procedure until $\alpha_{N}$ is expressed as a function of $T_{f}$. Now, every $\alpha_{i}$ has been expressed as a function of $T_{f}$. \\
	5.Apply the normalization equation:
	\begin{align*}
	\alpha_{0} + \alpha_{1} + ... + \alpha_{N} = 1
	\end{align*}
	to calculate the optimal finishing time $T_{f}$, as well as all the partitions $\alpha_{i}$s.
	\caption{Recursive algorithm to solve the optimal scheduling problem in a time-varying single level tree system}
\end{algorithm}
\section{Stochastic Analysis}
In the previous two subsections, two recursive algorithms to solve the optimal load fraction in the time-varying system were studied. However, the assumption that the time points of arrival and departure of background jobs are known a priori may not hold for many realistic circumstances. As a result, it is necessary to perform a more general analysis where the time points of arrival and departure of background jobs remain unknown.\par
In this section, we establish a stochastic model where the time points of arrival and departure of background jobs are not exactly known. To model the system, we assume Markovian statistics for the arrival and departure processes. Similar to the nature of arriving customers, the arrivals of background jobs are modeled as a Poisson random process with parameter $\lambda$ and the stay time for each of the background job followed an negative exponential distribution with parameter $\mu$. In this way, the system can be modeled as a M/M/1 queue. In [15], the average number of customers in the M/M/1 chain is used as the average number of the background jobs in each processor. However, this method may not be accurate given that the starting state and processing time are not taken into consideration. Also, [15] assumed that the system parameters $\lambda$ and $\mu$ were known, which may also not be possible. To deal with these issues, we first perform an estimation of $\lambda$ and $\mu$ based on the previous information of the system using a fading memory window. Then a simulation-based method is introduced to solve the optimal scheduling problem. In order to simplify and accelerate, an iterative algorithm is studied to achieve much faster running time with a sacrifice of negligible precision.\par
In this section the discussion is in the context that all the processors are time-varying (section 2.3) but this algorithm can work for both cases in section 2.2 and 2.3. Also, since we assume that the other transmissions have the same effect as the background jobs, we will just focus on the background jobs (processing speed) since the results also works for the other transmissions (communication speed). The numerical tests show that our stochastic model outperforms the method in [15].
\subsection{System Parameter Estimation}
In our system we assume that for any time-varying processor, the arrivals of background jobs follows a Poisson random process with parameter $\lambda$ and the stay time for each background job follows an exponential distribution. As a result, the arrivals and departures of background jobs form a M/M/1 queuing model. To this end, let ${x_1,x_2,x_3,....,x_n}$ be the samples of background jobs' inter-arrival interval times within the fading memory window. The fading memory window contains $n$ nearest samples before the divisible load job arrives, and the samples that are closer to the end point will receive a higher weight in the estimation.  As a result, the fading memory estimation will deliver a more stable result once the parameter varies with the time, otherwise it  will be just same as the normal estimation. These $n$ samples should be independent and identically distributed with $Exp(\lambda)$. To estimate the value of $\lambda$, the weighted maximum likelihood estimation (WMLE) method is used:
\begin{align}
lik(\lambda) = \prod_{i=1}^{n}(\lambda e^{-\lambda x_i})^{\beta_i}\\
\hat{\lambda} = \argmax_\lambda log(lik(\lambda))
\end{align}
where the $\beta_1,\beta_2,...,\beta_n$ are the fading memory weights with an ascending order. By solving the WMLE, the estimate of $\lambda$ can be achieved:\\
\begin{align}
\hat{\lambda} = \frac{\sum_{i=1}^{n}\beta_{i}}{\sum_{i=1}^{n}\beta_i x_i}
\end{align}
For the estimation of $\mu$, let ${y_1,y_2,...,y_n}$ be the samples of background stay time within the fading memory window. By applying the same method, the estimate of $\mu$ can be achieved as:\\
\begin{align}
\hat{\mu} = \frac{\sum_{i=1}^{n}\alpha_{i}}{\sum_{i=1}^{n}\alpha_i y_i}
\end{align}
where the $\alpha_1,\alpha_2,...,\alpha_n$ are the fading memory weights for $\mu$.
\subsection{Stochastic Model}
To solve the optimal scheduling using the stochastic model, we first introduce a simulation-based method. We take the median of a large number of samples to approximate the real case. Then a simplified iterative algorithm is introduced to reduce running time.
\subsubsection{Simulation-based Approach}
In the case where the actual arrival and departure times of background jobs are not known, it is impossible to make accurate schedule for the system since the real $W(t)$ and $Z(t)$ can never be obtained. To this end, a proper approximation is necessary for scheduling. Since the arrivals and departures of background jobs are modeled as a M/M/1 queue, it is naturally to gather statistic information from the M/M/1 queue with proper system parameter. In [15], given the system parameter $\lambda_i$ and $\mu_i$ for $ith$ processor, the average number of background jobs $n_i$ in the M/M/1 system can be estimated by $\frac{\rho_i}{1-\rho_i}$ where the $\rho_i = \frac{\lambda_i}{\mu_i}$. Then the average inverse of the processing speed was model as $\bar{W_i} = (n_i + 1)W_i$ since every background job is assume to share the equal computing power in [15]. In this way, the schedule can be achieved by solving equations (4) or (6). However, in the real case the average number of background jobs for processor $i$ during its processing time may not simply equal to the average state for the M/M/1 model due to two reasons. First, the processor may already have some background jobs be processed at the time when the divisible load job of our interest arrives, which means the start state of the M/M/1 model is not zero. Also, the average number of background jobs of a certain processor during its processing time may depend on how much time it takes to process. The divisible load job may terminate before the the M/M/1 queue reaches its equilibrium, so the average number of background jobs may not equal to the average number in the M/M/1 queue. \par
To deal with this issue, instead simply using the average number of background jobs as an approximation, a simulation based method is introduced in this paper. The main idea is to simulate background sequence for each processor, then the deterministic algorithm I or II can be applied. By operating this simulation for abundant times, the trial which achieves statistical median of the finishing time can be taken as the final schedule. The simulation of background jobs is based on the natural properties of M/M/1 queue: the time to stay in one state is a random variable with $Exp(\lambda + \mu)$ (except for the first state, which is $Exp(\lambda)$ since there is no departure), and the probability to move to the next largest state is $p_{next}=\frac{\lambda}{\lambda + \mu}$. Given the starting state $N_0$ and system parameters $\lambda$ and $\mu$ for each processor, the details of simulating M/M/1 based background sequence is described in Algorithm III.\par
\begin{algorithm}
	\begin{algorithmic}[1]
		\Require $N_0$, $\lambda$ and $\mu$
		\Ensure  Background sequence
		\State Set $t=0$;		
		\State Set M/M/1 state equals to $N_0$ at $t=0$;
		\Comment{The state represents the number of background jobs}
		\While{$t<T_f$}
		\If{current state equals to $0$}
		\State Generate a random variable $\tilde{t} \sim Exp(\lambda)$;
		\State $t=t+\tilde{t}$, move the state to $1$
		\Else
		\State Generate a random variable $\tilde{t} \sim Exp(\lambda +\mu)$;
		\State $t=t+\tilde{t}$;
		\State Generate a random variable $p \sim Unif(0,1)$;
		\If{$p <= p_{next}$}
		\State Move the M/M/1 queue to the next state;
		\Else
		\State Move the M/M/1 queue to the previous state;
		\EndIf
		\EndIf
		\EndWhile
	\end{algorithmic}
	\caption{Algorithm to simulate the background sequence}
\end{algorithm}
The system parameters $\lambda$ and $\mu$ can be estimated by the estimation step in section 3.1. This simulation can be done beforehand and the results stored in a table for future use. Based on the background jobs sequence and the pre-defined hypervisor function to assign a processor computation/communication power, $W(t)$ and $Z(t)$ can be achieved. Then the recursive deterministic algorithm I or II can be applied to obtain a schedule. By repeating this procedure for abundant times, various schedule plans can be achieved. The trial that achieve the median of all the finishing times is chosen as the final stochastic schedule plan.
\subsubsection{Iterative Algorithm for Simplification}
One drawback of the simulation-based algorithm is that it requires to run the recursive deterministic algorithm for abundant times. This procedure may become quite time-consuming when the system scale grows large since the recursive deterministic algorithm could be quite slow when the number of processors grows large. The running time can be significantly decreased if we can solve the linear equations (4) or (6) directly. However, solving the linear equations (4) or (6) requires the prior knowledge of $\bar{W}$ and $\bar{Z}$ for each processor, which can only be accessed after scheduling based on theorem I and II.\par 
To deal with this issue, an incorrect initial guess of the scheduling is made. This initial guess can be achieved either from the time-invariant approach or the result generated by [15]. After we achieve the initial scheduling, random background sequences are generated for each processor using algorithm III. Same as in the last subsection, $W(t)$ and $Z(t)$ can be estimated. Since we already have the initial schedule, we know the starting processing time of each processor. Based on theorem I and II $\bar{W}$ and $\bar{Z}$ can be estimated for each processor. An updated schedule can be achieved from solving linear equations (4) or (6). Then the background jobs sequences are generated again for each processor, and the updated $\bar{W}$ and $\bar{Z}$ can be achieved based on the new background jobs sequences. The updated $\bar{W}$ and $\bar{Z}$ could be utilized to update the schedule again. Similar as the previous subsection, abundant iterations of this procedure are performed and the trial that achieves the median of all the finishing times is chosen as the final stochastic schedule plan. An simplified algorithm description is shown in algorithm IV.\par  
\begin{algorithm}
	1. Perform an initial scheduling. The communication and processing time for each processor can be obtained.\\
	2. Run the Algorithm III to generate random background sequences for each processor.\\  
	3. Achieve the updated $\bar{W}_i$ and $\bar{Z}_i$ for each processor $i$ based on theorem I and II.\\
	4. Updating the schedule based on the new $\bar{W}$ and $\bar{Z}$.  The updated communication and processing time for each processor can be obtained.\\
	5. Repeat step 2 to 4 for an abundant number of times.\\
	6. The trial that achieve the median of all the finishing times is chosen as the final stochastic schedule plan.\\
	\caption{Simplified Scheduling}
\end{algorithm}  
Due to that each time the $\bar{W}$ and $\bar{Z}$ are estimated from the information of the last iteration, the overall scheduling may not as accurate as the simulation based method introduced in last subsection. However, the time-saving property of this method plays an important role when the system scale grows large. Numerical tests shows that for the system with large number of processors, this simplified iterative method can save significant time with negligible errors.  
\section{Numerical Test and Evaluation}
In this section we perform numerical tests for both deterministic and stochastic models. The first two subsections illustrates our results for the deterministic model using Algorithm I and II. We simulate each of the two algorithms in 50 time units and each time unit contain 100 time slots. That is to say, each time slot is equivalent to 0.01 unit of time. Usually the total process is finished within 50 time units. In these two subsections we use a simple way to generate the number of background jobs such that it is easier to perform evaluation of the system. A certain number of background jobs are generated throughout the 50 time units. The arrivals and departures of background jobs are simulated as uniformly distributed random time points in pairs and the departure time of a certain background jobs must be later than the arrival time. For simplicity we assume that the hypervisor evenly distributes the physical processor's computation/communication power among the virtual processors, which means that $W^h_i(n) = nW_i$. As a result, $W_{i}(t)$ and $Z_i(t)$ can be obtained based on the pre-defined $W_{i}$ and $Z_i$. For Algorithm I, $Z_i(t) = Z_i$ and $W_{1}(t) = W_{1}$ are set to be constant since the control processor is not time-varying. In our test, we arrange the processors' sequence according to their speed. That is to say, the faster processors will receive load prior to the slower ones. Based on this concept, we set the inverse of communication speed $Z_i = 1 + 0.1(i-1), i = 1,2,3,...,N$ for the processor $i$. Also, the parameters are set as: $T_{cm} = 1$, $T_{cp} = 4$ throughout the whole numerical test.\par
The third subsection illustrates our results by the stochastic model. The background jobs are generated by an M/M/1 queuing model instead of the simple method. We also compare our result with the result in [15]. It shows that our stochastic result better matches the deterministic result in terms of statistics.  
\subsection{Time-varying System with A Time-invariant Control Processor}
\subsubsection{Solution and Verification}
In this subsection, the control processor is time-invariant while the work processors are all time-varying. The link speed is assumed to be time-invariant. The $W_{i}$ is set to be equal for all processors and denoted as $W$. The Algorithm I is solved by starting with an initial $T_{f}$, then changing the value of $T_{f}$ gradually until achieving a sum of all $\alpha$s that is enough equal to 1. In this case, since $\alpha_{0}$ must be smaller than $1$, by the equation (4a), $\alpha_{0} = \frac{T_{f}}{W_{0}T_{cp}}$, then $T_{f}$ must be smaller than $W_{0}T_{cp}$. So the $T_{f}$ is initialized with its upper-bound $W_{0}T_{cp}$ and is decreased by a step of a time slot to achieve the correct solution. \par
%\FloatBarrier
\begin{figure}[h!]
	\centering
	\includegraphics[width=0.5\textwidth]{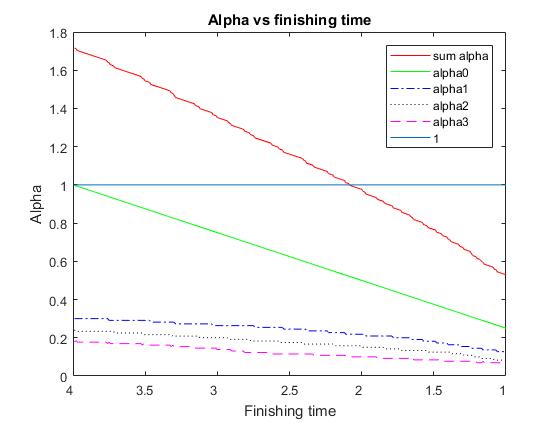}
	\caption{Finishing time vs the partitions of each processor by Algorithm I}
	\label{fig:a1_basis}
\end{figure}	
%\FloatBarrier
Fig. \ref{fig:a1_basis} shows how to achieve the optimal $T_f$ and all the partitions through Algorithm I. In this case, there are 3 worker processors and each worker processor has 40 background jobs for the whole 50 time units randomly generated. On average there are 0.8 background jobs each time unit for each processor. $W$ is set to be 1 and $T_{f}$ is initialized by $W_{0}T_{cp} = 4$ in this case. Our solution lays where the sum of alpha curve intersects with the line where the sum of the alphas equals to one.\par
\begin{table}[h]
	\caption{Two closest solution points for Fig. \ref{fig:a1_basis}}
	\begin{center}
		\begin{tabular}{||c c c c c c||} 
			\hline
			$T_{f}$ & $\alpha_{0}$ & $\alpha_{1}$ & $\alpha_{2}$ & $\alpha_{3}$ & sum \\ [0.5ex] 
			\hline\hline
			2.0800 & 0.5200 & 0.2182 & 0.1583 & 0.1000 & 0.9965 \\ 
			\hline
			2.0900 & 0.5225 & 0.2182 & 0.1583 & 0.1077 & 1.0067 \\ [1ex] 
			\hline
		\end{tabular}
	\end{center}
\end{table}
Table I shows the two closest solution points for Fig. \ref{fig:a1_basis}, where the sum of alphas is closest to $1$. From the solutions we can find that the sequence of divisible loads each processor takes is $\alpha_{0} > \alpha_{1} > \alpha_{2} > \alpha_{3}$, that is because the processor with the smaller index finishes communication before the one with larger index. In other words, the processor with smaller index has more time to process the loads. However in general the inequality part $\alpha_{1} > \alpha_{2} > \alpha_{3}$ does not always hold. Since the background jobs are generated randomly over the whole time interval, so it is possible that the processor with a larger index has less background jobs than the processor with smaller index during the processing time. Taking less background jobs means processing in a higher average speed. As a result, even with less time to process the loads, the processor with a larger index is possible to take more loads due to its fast speed. Especially, processor $0$ (control processor) would always take the majority part of the loads since it does not require communication and it always has a higher processing speed than other processors, because there is no background job on $P_{0}$. Either one of these two points can be taken as the solution of Algorithm I, one can also average these 2 points to achieve the solution. \par
To verify the accuracy of Algorithm I, we use Algorithm I to solve a time-invariant case where there is no background job and compare the result with the solution generated from equations (1a) to (1c) using the same parameters mentioned before.\par
\begin{table}[h]
	\caption{Solutions of equation (1)}
	\begin{center}
		\begin{tabular}{||c c c c c c||} 
			\hline
			$T_{f}$ & $\alpha_{0}$ & $\alpha_{1}$ & $\alpha_{2}$ & $\alpha_{3}$ & sum \\ [0.5ex] 
			\hline\hline
			1.4070 & 0.3517 & 0.2759 & 0.2122 & 0.1602 & 1.0000 \\ 
			\hline
		\end{tabular}
	\end{center}
\end{table}
\begin{table}[h]
	\caption{Closest solution point by Algorithm I without background job}
	\begin{center}
		\begin{tabular}{||c c c c c c||} 
			\hline
			$T_{f}$ & $\alpha_{0}$ & $\alpha_{1}$ & $\alpha_{2}$ & $\alpha_{3}$ & sum \\ [0.5ex] 
			\hline\hline
			1.4100 & 0.3525 & 0.2755 & 0.2117 & 0.1600 & 0.9996 \\ 
			\hline
		\end{tabular}
	\end{center}
\end{table}
Table II shows the solution of equation (1a) to (1c) while table III is the closest point by Algorithm I. One can find that the solution by Algorithm I matches the solution of equation (1a) to (1c) well.
\subsubsection{System Evaluation}
Two criteria are used to evaluate the time-varying system with a time-invariant control processor: finishing time and speedup. We will see how the system performs via these two criteria with a changing number of processors and background jobs. When the number of processors is changing, the number of background jobs is set to be 40 for each worker processor throughout 50 time units and when the number of background jobs is changing, the number of processors is set to be 4 (including the control processor). The definition of speedup will be introduced in the latter part of this subsection. In this subsection for each certain number of background jobs or processors, we run the Algorithm I 1000 times and average these trials to get a stable result.\par
%\FloatBarrier
\begin{figure}[h!]
	\centering
	\begin{subfigure}[b]{0.5\textwidth}
		\includegraphics[width=1\linewidth]{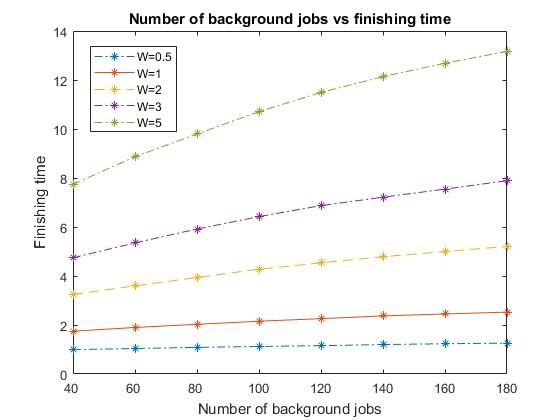}
		\caption{}
		\label{fig:a1_1} 
	\end{subfigure}
	\begin{subfigure}[b]{0.5\textwidth}
		\includegraphics[width=1\linewidth]{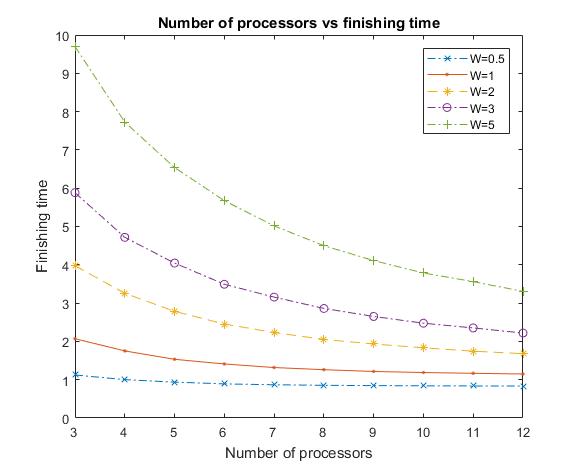}
		\caption{}
		\label{fig:a1_2}
	\end{subfigure}
	\caption[]{For a time varying system with time-invariant control processor (a) Number of background jobs vs finishing time (b) Number of processors vs finishing time.}
\end{figure}
%\FloatBarrier
Fig. \ref{fig:a1_1} shows how the finishing time varies with a increasing number of background jobs. One can find that the finishing time increases as the number of background jobs increases, which makes sense since more background jobs means less allocated to to the main job for a certain processor. Also, recall that all processors share a same inverse of processing speed $W$ when there is no background job, it is obvious that a higher $W$ (means lower speed) will make the system finish the job slower. This is also shown in the both Fig. \ref{fig:a1_1}  and Fig. \ref{fig:a1_2}, where higher $W$ has higher finishing time. Fig. \ref{fig:a1_2} shows the number of processors vs finishing time when each worker processor has 40 background jobs in total. With more processors sharing the same amount of job, the job should be finished faster, as shown in Fig. \ref{fig:a1_2}.\par
Since parallelism can accelerate the processing, one may wonder how much faster the parallel system can be compared with the sequential system. Defined by the well known Amdahl's law [21,22], speedup is the ratio of sequential processing time to parallel processing time for the same amount of load, which is:
\begin{align}
Speedup = \frac{T_{fs}}{T_{fp}}
\end{align}
Where $T_{fs}$ is the finishing time with a single processor while $T_{fp}$ is the finishing time with multiple parallel processors. The speedup can reflect how much faster the parallel system is compared with the sequential system. By taking the control processor as the single sequential processor, $T_{fs}$ can be achieved by:
\begin{align}
T_{fs} = 1WT_{cp} = WT_{cp}
\end{align}
As defined in equation (12) and (13), the speedup should have a positive correlation with the number of processor. Increasing of the number of processors means an increase of the parallelism in the system, which will result in a higher speedup value. \par
\begin{figure}[h!]
	\centering
	\includegraphics[width=1\linewidth]{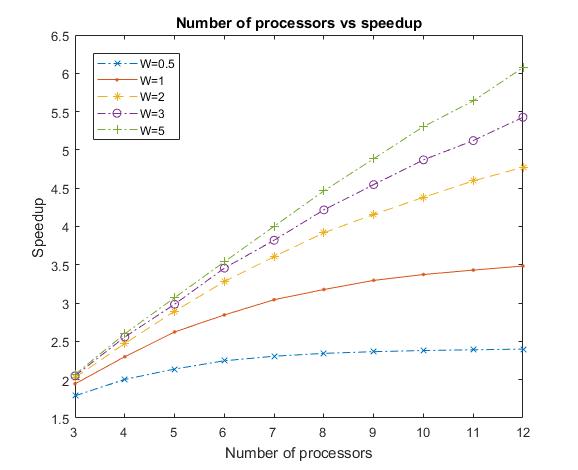}
	\caption[]{ For a time varying system with time-invariant control processor: number of processors vs speedup.}
	\label{fig:a1_4}
	%\label{fig:a1_34}
\end{figure}
To verify our expectations, Fig. \ref{fig:a1_4} demonstrates how speedup varies with the number of processors. This matches our expectation. For the relationship between $W$ and speedup, from the figure we can find that the higher $W$ results in higher speedup. This is because the $T_{fs}$ is linear to $W$, which should be more sensitive than $T_{fp}$ to $W$. In other words, $T_{fs}$ changes faster than $T_{fp}$ when $W$ changes. Then, for a certain number of processors, a higher $W$ will result in a higher speedup. In other words, parallelism has a bigger benefit for the slower system.
\subsection{Time-varying System with Time-varying Control Processor, Processing and Communication Speed}
\subsubsection{Solutions and Verification}
In this subsection, there are background jobs at the control processors as well. Furthermore, there will be interfering communications, which will make the control processor have both time-varying processing speed and communication speed. The number of extra connections in control processor is set to be equal to the number of background jobs in this processor for the whole time interval. $Z(t)$ is generated in the same manner as $W(t)$ described at the beginning of this section.\par
%\FloatBarrier
\begin{figure}[h!]
	\centering
	\includegraphics[width=0.5\textwidth]{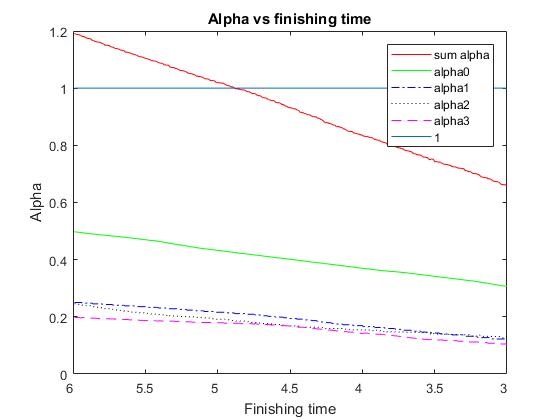}
	\caption{Finishing time vs the partitions of each processor by Algorithm II}
	\label{fig:a2_basis}
\end{figure}	
%\FloatBarrier
\begin{table}[h]
	\caption{Two closest solution points for Fig. \ref{fig:a2_basis}}
	\begin{center}
		\begin{tabular}{||c c c c c c||} 
			\hline
			$T_{f}$ & $\alpha_{0}$ & $\alpha_{1}$ & $\alpha_{2}$ & $\alpha_{3}$ & sum \\ [0.5ex] 
			\hline\hline
			4.8900 & 0.4244 & 0.2113 & 0.1849 & 0.1774 & 0.9980\\ 
			\hline
			4.9000 & 0.4257 & 0.2137 & 0.1866 & 0.1777 & 1.0037\\ [1ex] 
			\hline
		\end{tabular}
	\end{center}
\end{table}
Similar to the previous subsection, Fig. \ref{fig:a2_basis} and Table IV shows how to get the solution using Algorithm II. There are one control processor and three worker processors and each processor has 40 background jobs. The control processor also has 40 other incoming and outcoming network connections. Again, all processors share the same processing speed when there is no background job, as $W=1$ for all.\par
\begin{table}[h]
	\caption{Closest solution point by Algorithm II without background job}
	\begin{center}
		\begin{tabular}{||c c c c c c||} 
			\hline
			$T_{f}$ & $\alpha_{0}$ & $\alpha_{1}$ & $\alpha_{2}$ & $\alpha_{3}$ & sum \\ [0.5ex] 
			\hline\hline
			1.4110 & 0.3528 & 0.2755 & 0.2117 & 0.1600 & 0.9999 \\ 
			\hline
		\end{tabular}
	\end{center}
\end{table}
The same method as the previous subsection is applied to verify our program. The result is shown in Table V. Again our solution matches the solution in Table II. 
\subsubsection{System Evaluation}
The same two criteria: finishing time and speedup are used to evaluate the time-varying system with time-varying control processor, processing and communication speed. Again we change the number of processors and background jobs to see how the system performs. The number of background jobs is set to be 40 and the number of processors is set to be four (one control processor and three worker processors) when the other one is changing. Algorithm II is also averaged over 1000 trails for a stable result.\par
Fig. \ref{fig:a2_12} shows how the finishing time varies with a increasing number of processors and background jobs for three different $W$ values. As the previous subsection, finishing time has a positive correlation with the number of background jobs but negative correlation with the number of processors for the same reason. \par
\begin{figure}[!h]
	\centering
	\begin{subfigure}[b]{0.5\textwidth}
		\includegraphics[width=1\linewidth]{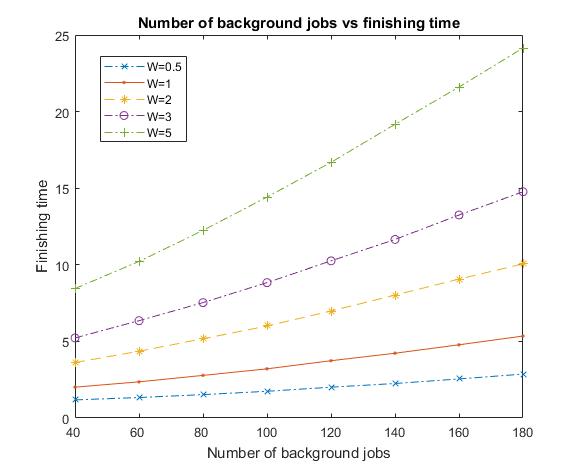}
		\caption{}
		\label{fig:a2_1} 
	\end{subfigure}
	\begin{subfigure}[b]{0.5\textwidth}
		\includegraphics[width=1\linewidth]{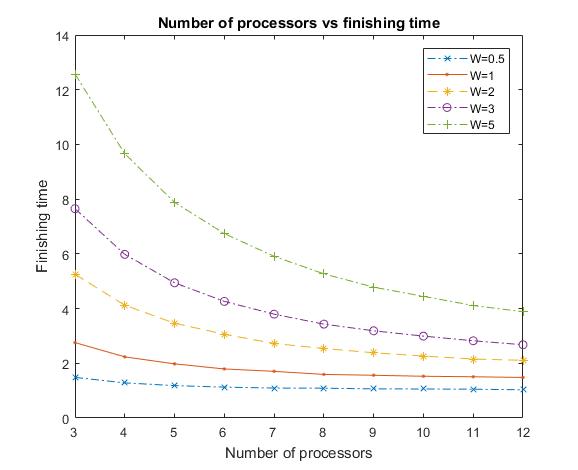}
		\caption{}
		\label{fig:a2_2}
	\end{subfigure}
	\caption[]{For a time-varying system with time-varying control processor (a) Number of background jobs vs finishing time (b) Number of processors vs finishing time.}
	\label{fig:a2_12}
\end{figure}
In case of Speedup, it is more complicated since now our reference single sequential processor $P_{0}$ is also time-varying. Equation (12) is still used to define Speedup, but equation (13) cannot achieve $T_{fs}$ for this case. To solve this problem, by taking $\alpha_{0}$ as $1$ in equation (6a), $T_{fs}$ can be obtained by solving the following equation:
\begin{IEEEeqnarray}{rCl}
	T_{fs} = 1\bar{W}_{0}(t)T_{cp} = \frac{T_{fs}}{\int_{0}^{T_{fs}}\frac{1}{W_{0}(t)}dt}T_{cp} \nonumber\\
	\Longrightarrow
	\frac{1}{T_{cp}}\int_{0}^{T_{fs}}\frac{1}{W_{0}(t)}dt = 1
\end{IEEEeqnarray}
Fig. \ref{fig:a2_4} demonstrates the relationship between speedup and the number of processors.  One can see that speedup will increase as the number of processors increases. This is similar to Fig. \ref{fig:a1_4} and also meets our expectation. 
\begin{figure}[h!]
	\centering
	\includegraphics[width=1\linewidth]{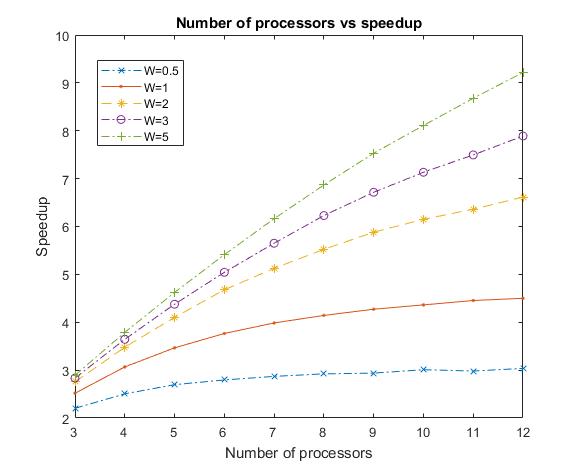}
	\caption[]{For a time varying system with time-varying control processor: number of processors vs speedup.}
	\label{fig:a2_4}
	%\label{fig:a1_34}
\end{figure}
\subsection{Stochastic Model}
In this subsection the background jobs are generated by a M/M/1 model. The generation is similar to the method described in section 3.2. The starting state for each processor is taken to be zero for simplicity. Both cases where the control processor is time-varying or time-invariant are tested. In the test, we call the result generated by the simulation-based method as the \enquote{simulation-based}. We also note the result provided by [15] as \enquote{before correction} and our simplified iterative Algorithm as \enquote{iterative}. A result for 4 processors (one control processor and three worker processors) are shown in Fig. \ref{fig:st1}.\par
\begin{figure}[!h]
	\centering
	\begin{subfigure}[b]{0.5\textwidth}
		\includegraphics[width=1\linewidth]{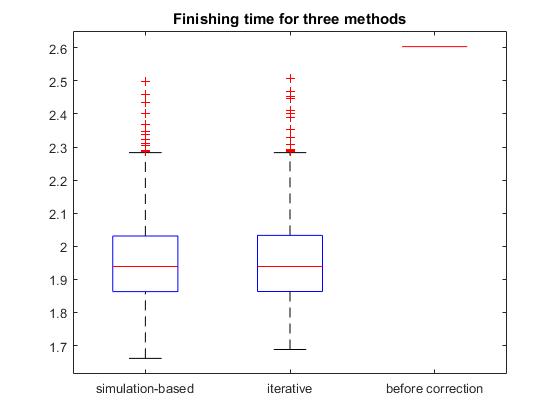}
		\caption{}
		\label{fig:Ng9} 
	\end{subfigure}
	\begin{subfigure}[b]{0.5\textwidth}
		\includegraphics[width=1\linewidth]{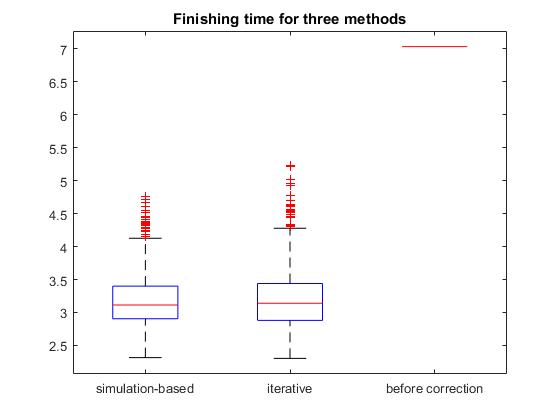}
		\caption{}
		\label{fig:Ng10}
	\end{subfigure}
	\caption[]{Result for 4 processors (a) Time-invariant control processor (b) Time-varying control processor, processing and communication speed.}
	\label{fig:st1}
\end{figure}
In Fig. \ref{fig:st1} we use box-plotted finishing time as the criterion to compare the three methods. The details of box plot can be found at [26]. Briefly speaking the box contains 50\% of the data, whose lower and upper boundary lines are at the 25\%/75\% quantile of the data. A central line indicates the median of the data and some outliers of data are plotted out as dots. The median of each data is picked as the stochastic solution. In this case, the starting states for all processors are set to be zero homogeneously. Here $\lambda$ is set to be 0.1 and $\mu$ is set to be 0.125. Based on these settings, the scheduling will be finished before the M/M/1 queue reaches its average state number in general. As a result, \enquote{before correction} method will deliver a higher finishing time since that \enquote{before correction} method will get an incorrect higher number of background jobs which will results in a slower processing speed in general. From the figure we can find that our \enquote{iterative} method delivers similar result as the \enquote{deterministic} result in statistics, which is lower than the \enquote{before correction} method for both cases whether the control processor is time-varying or not. In this case there is only 4 processors and the times to run the two algorithms are quite close. The simulation-based method turns out to be a better solution than the simplified iterative one since it is more accurate. Another case with more processors in shown in Fig. \ref{fig:st2}. \par
\begin{figure}[!h]
	\centering
	\begin{subfigure}[b]{0.5\textwidth}
		\includegraphics[width=1\linewidth]{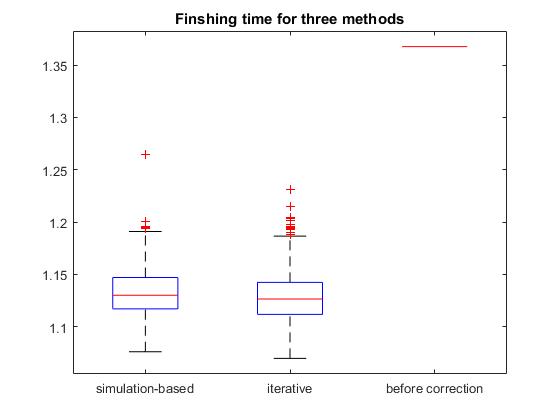}
		\caption{}
		\label{fig:s30_1} 
	\end{subfigure}
	\begin{subfigure}[b]{0.5\textwidth}
		\includegraphics[width=1\linewidth]{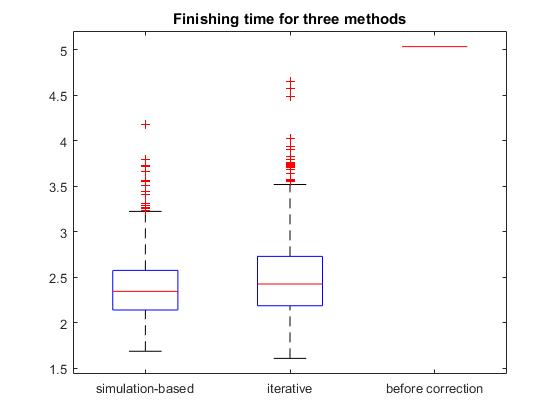}
		\caption{}
		\label{fig:s30_2} 
	\end{subfigure}
	\caption[]{Result for 15 processors (a)Time-invariant control processor (b)Time-varying control processor, processing and communication speed.}
	\label{fig:st2}
\end{figure}
Fig. \ref{fig:st2} demonstrates the result for 15 processors (one control processor and 14 worker processors). The performance of the simplified iterative algorithm is similar as the case when there is only 4 processors. However in this case, running the simulated based algorithm is overwhelmingly time-consuming (2213 seconds for Fig. \ref{fig:s30_1} and 2418 seconds for Fig. \ref{fig:s30_2}) while the simplified iterative algorithm can be time-saving (182 seconds for Fig. \ref{fig:s30_1} and 267 seconds for Fig. \ref{fig:s30_2}).\par 
One thing to note is that once the system parameters $\lambda$ and $\mu$ are fixed, the only factor that influences the stochastic model are the starting state of each processor. This is totally different from the deterministic case, which is dependent on the real distribution of the background jobs on each trial.
\section{Conclusion}
This paper studied optimal divisible loads scheduling of time-varying single level tree network. The time-varying processing speed and channel speed were transformed into equivalent time-invariant ones. The deterministic analysis was first studied where the arrival and departure times are known To achieve the optimal partition for each processor, two recursive algorithms were developed in case whether the control processor is time-invariant or time-varying. For stochastic analysis, the arrival and departure of background jobs are modeled as a M/M/1 queuing model and two algorithms are provided to solve the scheduling problem. Extensive numerical tests were performed to demonstrate the relationships between finishing time, speedup, background job number and processor number.\par
Future enhancement for this research can be pursued under the context of various network topologies such as multi-level tree or mesh. Also, the system model can be extended to handle more complicated cases, such as a general distribution of the arrivals and departures of background jobs/transmissions in stochastic analysis. 

\ifCLASSOPTIONcaptionsoff
  \newpage
\fi



%\bibliographystyle{IEEEtran}

\begin{IEEEbiography}[{\includegraphics[width=1in,height=1.25in,clip,keepaspectratio]{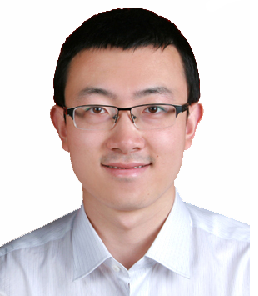}}]{Fei Wu}
received the BE degree in information and telecommunication engineering from Xi'an Jiaotong University, Xi'an, China, in 2012, and the MS degree in electrical engineering from Stony Brook University, Stony Brook, New York, in 2013. He is currently working toward the PhD degree in electrical engineering at Stony Brook University. His research interests include scheduling, parallel processing, computer networks and virtualization. 
\end{IEEEbiography}
\begin{IEEEbiography}[{\includegraphics[width=1in,height=1.25in,clip,keepaspectratio]{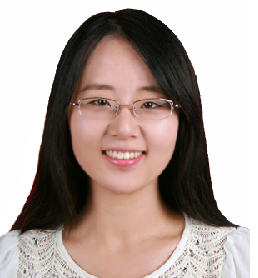}}]{Yang Cao}
received the BE degree in Electrical Engineering and Automation from Northwestern Polytechnical University, Xi'an, China, in June 2012. She also received MS degree in Electrical Engineering from Stony Brook University, Stony Brook, New York, in December 2013. Currently she is working toward the PhD degree in Electrical Engineering at Stony Brook University. Her research interests include task scheduling and resource allocation in distributed systems, cloud networks,  data centers, etc.
\end{IEEEbiography}
\begin{IEEEbiography}[{\includegraphics[width=1in,height=1.25in,clip,keepaspectratio]{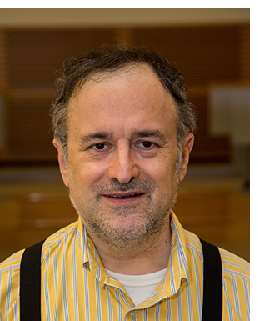}}]{Thomas G. Robertazzi}
received the BEE degree from Cooper Union, New York, in 1977 and the PhD degree from Princeton University, Princeton, New Jersey, in 1981. He is presently a professor in the Department of Electrical and Computer Engineering, Stony Brook University, Stony Brook, New York. He has published extensively in the areas of parallel processing scheduling, telecommunications and performance evaluation. He has also authored, co-authored or edited six books in the areas of networking, performance evaluation, scheduling and network planning. He is a fellow of the IEEE and since 2008 co-chair of the Stony Brook University Senate Research Committee.
\end{IEEEbiography}


%\vfill


\end{document}